\documentclass[reqno]{amsart}
\usepackage{amsfonts}
\usepackage{lmodern}
\usepackage[english]{babel}
\usepackage[utf8]{inputenc}
\newtheorem{theorem}{Theorem}
\theoremstyle{plain}

\newtheorem{definition}{Definition}

\newtheorem{proposition}{Proposition}

\numberwithin{equation}{section}
 \numberwithin{theorem}{section}
 \numberwithin{proposition}{section}
 \numberwithin{remark}{section}
 \numberwithin{definition}{section}
 \numberwithin{lemma}{section}
 \numberwithin{corollary}{section}
 \numberwithin{example}{section}
 \numberwithin{claim}{section}

\begin{document}
\title[Homogenization of variational inequalities]{%
Deterministic homogenization of variational inequalities  with unilateral constraints}
\author{Hermann Douanla}
\address{Hermann Douanla, Department of Mathematics, University of Yaounde
1, P.O. Box 812, Yaounde, Cameroon}
\email{hdouanla@gmail.com}
\author{Cyrille Kenne}
\address{Cyrille Kenne, Department of Mathematics, University of Yaounde
1, P.O. Box 812, Yaounde, Cameroon}
\email{kennestevec@gmail.com}
\date{November, 2018}
\subjclass[2010]{76M50, 35B40, 35B27,  35Q35 }
\keywords{deterministic homogenization, variational inequality, unilateral constraint, multiscale convergence}

\begin{abstract}
The article studies the reiterated homogenization of linear elliptic variational inequalities arising in problems with unilateral constrains. We assume that the coefficients of the equations satisfy and abstract hypothesis covering on each scale a large set of concrete deterministic behavior such as the periodic, the almost periodic, the convergence at infinity. Using the multi-scale convergence method, we derive a homogenization result whose limit problem is of the same type as the  problem with rapidly  oscillating coefficients.
\end{abstract}

\maketitle

\section{Introduction\label{S1}}
Let $\Omega$ be a bounded domain in $\mathbb{R}^N$ (integer $N\geq 2$) locally located on one side of its Lipschitz boundary $\partial\Omega$. Let $(\psi_\varepsilon)_{\varepsilon>0}$ be a sequence of functions in   $\{ u\in H^1(\Omega)$, $\gamma_0(u) \leq 0 \ \mbox{sur} \ \partial\Omega\}$ (where $\gamma_0 $ denotes the zero order  trace operator on $\partial\Omega$) that satisfies
\begin{equation}\label{eq001}
\psi_\varepsilon \to \psi_0\quad \text{ in } H^1(\Omega)\ \ \ \text{ as }\ \varepsilon\to 0.
\end{equation}
For fixed $\varepsilon>0$, we define 
$$
K_\varepsilon=\{v \in H^1_0(\Omega);\ \ v \geq \psi_\varepsilon \ \ \mbox{a.e., in  }\Omega\}.
$$
 Given $f\in L^2(\Omega)$, we study the asymptotic behavior (as $0<\varepsilon\to 0$) of the solution 
 to the following variational inequality
\begin{eqnarray} \label{eq1.1}
            \left \{ \begin{array}{lll}
            \mbox{Find} \ u_\varepsilon \in \ K_\varepsilon \ \ \mbox{such that} \\ 
          \displaystyle{\int_{\Omega}A(x, \frac{x}{\varepsilon},\frac{x}{\varepsilon^2})\nabla u_\varepsilon(x)\cdot\nabla (v(x)- u_\varepsilon(x))dx} &\geq& \displaystyle{\int_{\Omega} f(x)(v(x)- u_\varepsilon(x))dx} \\
           \mbox{for all} \ v \in K_\varepsilon,     
              \end{array}\right.
            \end{eqnarray} 
where the matrix $A=(a_{ij})_{1\leq i,j\leq N}$ with rapidly oscillating coefficients satisfies the following hypotheses: for any  $1\leq i,j\leq N$, we have 
\begin{eqnarray}
&&a_{ij}\in L^2(\Omega;\mathcal{B}(\mathbb{R}^N_y\times\mathbb{R}^N_z)),\label{eq11.2}\\
&&a_{ij}=a_{ji}, \label{eq11.3}
\end{eqnarray}
 and there exist two positive real constants $\alpha>0$ and $\beta>0$ such that for almost all $x\in \Omega$ and for all $(y,z)\in \mathbb{R}^N_y\times\mathbb{R}^N_z$
\begin{equation}\label{eq11.4}
\alpha |\xi|^2  \leq  \sum_{i,j=1}^N a_{ij}\xi_i\xi_j \leq \beta |\xi|^2\qquad\qquad \forall\xi=(\xi_i)_{1\leq i\leq N}\in\mathbb{R}^N.
\end{equation}

The problems of this form are usualy refered to as obstacle or unilateral problems \cite{DL}. Indeed, the set $K_\varepsilon$ in made up of functions constrained to stay  on one side of and obstacle modeled by the function $\psi_\varepsilon$.  The problems of type (\ref{eq1.1}) ussually  appears when modelling phenomena in plasticity theory, unilateral contact mechanics,  economics, and engineering (\cite{DL,oden}). This type of problems
may also arise in thermal diffusion problems or in biology, in the modeling of chemical flows in cells surrounded by semipermeable membranes \cite{CT}. Since the pioneering work of Signorini \cite{signorini}, the mathematical theory of variational inequalities with unilateral constrains has attracted the attention of a huge number of researchers (\cite{Brezis, CMT, DL, fichera,  oden, evariste, TS} and references therein). The limiting behavior of solutions to variational inequalities with constraints has also been studied by many scientists, see for exemple \cite{Caffarelli, Focardi, Sandrakov, CT} and references therein.

But since the development of the theory of deterministic homogenization hinted in \cite{ZK} and developed in e.g.,  \cite{Casado, gabihom1, BJMA}, see also \cite{DNW, DW1, NSW, SSW} where it has been utilized, it hasn't been used to address the homogenization of variational inequalities. Our main result formulates as follows.

\begin{theorem}\label{t1.1}
For any $\varepsilon >0$, let $u_\varepsilon$ be the unique solution to the variational inequality (\ref{eq1.1}) with the hypotheses (\ref{eq001}), (\ref{eq11.2})-(\ref{eq11.4}) and (\ref{eq33.1}). Then as $0< \varepsilon \longrightarrow 0$, we have $u_\varepsilon \longrightarrow u_0$ in $H^1_0(\Omega)$-weak, and strongly in $L^2(\Omega)$, where $u_0$ is the unique solution to the following homogenized variational inequality
\begin{eqnarray*}
     \left \{ \begin{array}{lll} 
     u_0\in K_0,\\ 
  \displaystyle{   \int_{\Omega}A^{*}(x)\nabla u_0\cdot\nabla (v_0-u_0)\ dx  \geq \int_{\Omega}f(v_0-u_0) \ dx },\\
  \mbox{for all}\ v_0\in K_0,
   \end{array}\right.
   \end{eqnarray*} 
where,
$K_0=\{ H^1_0(\Omega), v_0\geq \psi_0  \ a.e \ \mbox{in} \ \Omega\}$
 and where, for almost every $x\in \Omega$, the homogenized matrix is given by
$$
A^{*}(x)=M[A(I+\nabla_z \chi)(I+ \nabla_y \theta)],
$$
the functions $\chi$ and $\theta$ being the solutions to the microscopic problem (\ref{eqt17}) and the mesoscopic problem (\ref{eqt22}), respectively, while $M$ denotes the mean value operator. 
\end{theorem}  

The paper is organized as follows: In Section\,2 we briefly present the concepts of algebra with mean value and that of multiscale convergence. Section\,3 deals with estimates and the passage to the limit while the main result is proved in Section\,4. Finally, for the sake of simplicity, vector spaces are consider over $\mathbb{R}$.
\section{Algebras with mean value and multiscale convergence}
In this section, we recall the concept of algebra with mean value \cite{BJMA,ZK} and that of multiscale convergence \cite{ AB96, gabi89, gabihom1, BJMA, SSW}. A detail treatment of the results in this section may be found in \cite{BJMA}.
\subsection{Algebra with mean value} Let $\mathcal{H} = (H_\varepsilon)_{\varepsilon>0}$ be the  action of $\mathbb{R}^*_+$ (the multiplicative group of positive real numbers) on the numerical space $\mathbb{R}^N$ defined by
\begin{equation}
H_\varepsilon(x)=\frac{x}{\zeta(\varepsilon)}\qquad (x\in\mathbb{R}^N)
\end{equation}
where $\zeta$ is a strictly positive function of $\varepsilon$ tending to zero together with $\varepsilon$. For $u\in L^2_{loc}(\mathbb{R}^N)$ and $\varepsilon>0$, we defined
\begin{equation}
u^\varepsilon(x)=u(H_\varepsilon(x))\qquad (x\in\mathbb{R}^N),
\end{equation}
a function lying in $L^2_{loc}(\mathbb{R}^N)$. A bounded uniformly continuous real-valued function on $\mathbb{R}^N$ possesses a mean value for $\mathcal{H}$ if the sequence $(u^\varepsilon)_{\varepsilon>0}$ weakly* converges in $L^\infty(\mathbb{R}^N)$. An algebra with mean value for $\mathcal{H}$ (algebra wmv, in short) on $\mathbb{R}^N$ is defined to be a closed subalgebra of the algebra of bounded uniformly continuous real-valued function on $\mathbb{R}^N$, BUC($\mathbb{R}^N$), which contains the constants, is translation invariant and is such that any of its element possesses a mean value. Now, Let $A$ be an algebra wmv. The mean value of $u\in A$, denoted hereafter by $M(u)$, writes as  
\begin{equation}
M(u)=\lim_{R\to +\infty}\frac{1}{|B_R|}\int_{B_R}u(y)dy
\end{equation} 
where $B_R$ stands for the bounded open ball in $\mathbb{R}^N$ centered at the origin and with radius $R$, and $|B_R|$ denotes its Lebesgue measure. Indeed, let $R$ be a positive number and set $\zeta(\varepsilon)=\frac{1}{R}$. Then $R\to +\infty$ as $\varepsilon\to 0$. Since $u^\varepsilon\to M(u)$ in $L^\infty(\mathbb{R}^N)$-weak*, we have $\int_{\mathbb{R}^N} u^\varepsilon\chi_{B_1}dx\to M(U)|B_1|$ as $R\to +\infty$, where $B_1$ denotes the unit ball in $\mathbb{R}^N$ and $\chi_{B_1}$ its characteristic function. But $\int_{\mathbb{R}^N} u^\varepsilon\chi_{B_1}dx=\int_{B_1}u(Rx)dx$, and a change of variable $y=Rx$ gives 
$$
\frac{1}{|B_1|}\int_{B_1}u(Rx)dx=\frac{1}{R^N|B_1|}\int_{B_R}u(y)dy = \frac{1}{|B_R|}\int_{B_R}u(y)dy,
$$
and our claim is justified.

We recall that given an algebra wmv $A$ and a positive integer $m\geq 0$, we define regular subalgebras of $A$ by setting
$$
A^m=\{\psi\in\mathcal{C}^m(\mathbb{R}^N) : D^\alpha \psi \in A,\  \forall \alpha=(\alpha_1,\cdots,\alpha_N)\in\mathbb{N}^N \ \mathrm{ with }\ |\alpha|\leq m\},
$$
where $D^\alpha\psi=\frac{\partial^{|\alpha|}\psi}{\partial_{y_1}^{\alpha_1}\cdots\partial_{y_N}^{\alpha_N}}$. For finite $m$ the norm $\||u|\|_m=\sup_{|\alpha|\leq m}\|D^\alpha\psi\|_\infty$, makes $A^m$ a Banach space. We also define 
$$
A^\infty=\{\psi\in\mathcal{C}^\infty(\mathbb{R}^N) : D^\alpha \psi \in A,\  \forall \alpha=(\alpha_1,\cdots,\alpha_N)\in\mathbb{N}^N \},
$$
a Fréchet space when endowed with the locally convex topology defined by the family of norms $\||\cdot|\|_m$. Moreover the space $A^\infty$  is dense in any $A^m$ ($m\in\mathbb{N}$).

As we are concerned in this work with reiterated homogenization, the notion of product algebra wmv will be useful. In this direction, we first recall some facts about  vector-valued algebra wmv.  Let $F$ be a Banach space. We recall that $\mathrm{BUC}(\mathbb{R}^N;F)$, the space of bounded uniform continuous functions $u:\mathbb{R}^N\to F$, is a Banach space when endowed with the following norm
$$
\|u\|_\infty= \sup_{y\in\mathbb{R}^N}\|u(y)\|_F \qquad \qquad (u\in \mathrm{BUC}(\mathbb{R}^N;F)),
$$
where $\|\cdot\|_F$ stands for the norm in F. Let $A$ be an algebra wmv, we denote by $A\bigotimes F$ the space of functions of the form
$$
\sum_{finite}u_i\otimes e_i,\quad  \ u_i\in A\  \mathrm{ and }\ e_i\in F,
$$
where the function $u_i\otimes e_i$ is defined by $(u_i\otimes e_i) (y)=u_i(y)e_i$ for $y\in\mathbb{R}^N$. This being so, we define the vector-valued algebra wmv $A(\mathbb{R}^N ; F)$ as the closure of $A\bigotimes F$ in $\mathrm{BUC}(\mathbb{R}^N;F)$.

We can now introduce the notion of product algebra wmv. Let $N_1$ and $N_2$ be two strictly positive integers and let $A_y$ and $A_z$ be two algebras wmv on $\mathbb{R}_y^{N_1}$  and $\mathbb{R}_z^{N_2}$, respectively. The product algebra wmv of $A_y$ and $A_z$ is denoted by $A_y\odot A_z$ and defines as the vector-valued algebra wmv $A_y(\mathbb{R}_y^{N_1};A_z)$ \ $(\ \equiv A_z(\mathbb{R}_z^{N_2};A_y) $. Obviously, $A=A_y\odot A_z$ is an algebra wmv on $\mathbb{R}^{N_1 + N_2}_{y,z}$.

We recall that the vector-valued Marcinkiewicz space $\mathfrak{M}^{2}(%
\mathbb{R}^{N};F)$ is defined as the set of functions $u\in L_{loc}^{2}(\mathbb{R}%
^{N};F) $ such that 
\begin{equation*}
\underset{R\rightarrow +\infty }{\lim \sup }\frac{1}{|B_R|}
\!\int_{B_{R}}\left\Vert u(y)\right\Vert_F^{2}dy<\infty .
\end{equation*}%
Endowed with the seminorm 
\begin{equation*}
\left\Vert u\right\Vert _{2,F}=\left( \underset{R\rightarrow +\infty }{\lim
\sup }%
\frac{1}{|B_R|}\!\int_{B_{R}}\left\Vert u(y)\right\Vert_F ^{2}dy\right) ^{1/2},
\end{equation*}
$\mathfrak{M}^{2}(\mathbb{R}^{N}, F)$ is a complete seminormed space containing $A(\mathbb{R}^N,F)$. Next, the \textit{generalized Besicovitch} \textit{space} $B_{A}^{2}(%
\mathbb{R}^{N})$ is defined to be the closure of the vector-valued algebra wmv $A(\mathbb{R}^N;F)$ in $\mathfrak{M}^{2}(\mathbb{R}^{N};F)$. Then the following holds true:
\begin{itemize}
\item[(i)] The mean value $M : A(\mathbb{R}^N;F)\to F$ extends by continuity to a continuous linear mapping (still denoted by $M$) on $B^2_A(\mathbb{R}^N;F)$ satisfying $ T(M(u))=M(T(u))$ for all $T\in F'$ and all $u\in B^2_A(\mathbb{R}^N;F)$. Moreover, for $u\in B^2_A(\mathbb{R}^N;F)$ we have
$$
\|u\|_{2,F}=\left[ M(\|u\|^2_F) \right]^\frac{1}{2}=\left[\lim_{R\to +\infty}\frac{1}{|B_R|}\int_{B_R}\|u(y)\|^2_F dy\right]^\frac{1}{2}
$$ 
\item[(ii)] Let $\mathcal{N}=\{u\in B^2_A(\mathbb{R}^N;F) : \|u\|_{2,F}=0\}$. Then the quotient space $\mathcal{B}^2_A(\mathbb{R}^N;F)=B^2_A(\mathbb{R}^N;F)/\mathcal{N}$ is a Banach space when endowed with the norm 
$$
\|u+\mathcal{N}\|_{2,F}=\|u\|_{2,F}\qquad \text{for }u\in B^2_A(\mathbb{R}^N;F).
$$
Moreover, if $F$ is a Hilbert space, then so is  $\mathcal{B}^2_A(\mathbb{R}^N;F)$ with inner product
\begin{equation}
(u,v)_2=M\left[ (u,v)_F \right]\qquad \text{for }\ \ u,v\in \mathcal{B}^2_A(\mathbb{R}^N;F).
\end{equation} 
\end{itemize}
The standard case $F=\mathbb{R}$ is of particular interest. In this case we simplify the notations and write $B^2_A(\mathbb{R}^N),\ \mathcal{B}^2_A(\mathbb{R}^N)$ and $\|\cdot\|_{2}$ in place of $B^2_A(\mathbb{R}^N;\mathbb{R}),\ \mathcal{B}^2_A(\mathbb{R}^N;\mathbb{R})$ and $\|\cdot\|_{2, \mathbb{R}}$, respectively. Also, we recall that the space 
\begin{equation*}
B_{A}^{1,2}(\mathbb{R}^{N})=\{u\in B_{A}^{2}(\mathbb{R}^{N}):\nabla _{y}u\in
(B_{A}^{2}(\mathbb{R}^{N}))^{N}\}
\end{equation*}%
endowed with the seminorm 
\begin{equation*}
\left\Vert u\right\Vert_{1,2}=\left( \left\Vert u\right\Vert_2\right)^{\frac{1}{2}},
\end{equation*}%
which is a complete seminormed space. Its Banach counterpart is defined as follows.
$$
\mathcal{B}_{A}^{1,2}(\mathbb{R}^{N})=\{u\in \mathcal{B}_{A}^{2}(\mathbb{R}^{N}): \overline{\nabla }_{y}u\in
(B_{A}^{2}(\mathbb{R}^{N}))^{N}\},
$$
where $\overline{\nabla}_y=(\overline{\partial}/\partial y_i)_{i=1,\cdots,N}$ and  $\overline{%
\partial }/\partial y_{i}$ is defined by 
\begin{equation}
\frac{\overline{\partial }}{\partial y_{i}}(u+\mathcal{N}):=\frac{\partial u%
}{\partial y_{i}}+\mathcal{N}\quad \text{ for }u\in B_{A}^{1,2}(\mathbb{R}^{N}).
\label{eq2.5}
\end{equation}
It is important to note that $\overline{\partial }/\partial y_{i}$ is also
defined as the infinitesimal generator in the $i$th direction coordinate of
the strongly continuous group $\mathcal{T}(y):\mathcal{B}_{A}^{2}(\mathbb{R}%
^{N})\rightarrow \mathcal{B}_{A}^{2}(\mathbb{R}^{N});\ \mathcal{T}(y)(u+%
\mathcal{N})=u(\cdot +y)+\mathcal{N}$. Let us denote by $\varrho :B_{A}^{2}(%
\mathbb{R}^{N})\rightarrow \mathcal{B}_{A}^{2}(\mathbb{R}^{N})=B_{A}^{2}(%
\mathbb{R}^{N})/\mathcal{N}$, $\varrho (u)=u+\mathcal{N}$, the canonical
surjection. We remark that if $u\in B_{A}^{1,2}(\mathbb{R}^{N})$ then $%
\varrho (u)\in \mathcal{B}_{A}^{1,2}(\mathbb{R}^{N})$ with (see (\ref{eq2.5}))
\begin{equation*}
\frac{\overline{\partial }\varrho (u)}{\partial y_{i}}=\varrho \left( \frac{%
\partial u}{\partial y_{i}}\right).
\end{equation*}

As pointed out in \cite[Remark\,2.13]{NSW}, if there exist translation invariant  elements in the algebra wmv that are not constants, the result obtained after the homogenization process might be useless. 

\begin{definition} An algebra wmv $A$ on $\mathbb{R}^N$ is said to be ergodic if any $u\in \mathcal{B}^2_A(\mathbb{R}^N)$ such that $\mathcal{T}(y)u=u$ for every $y$ in $\mathbb{R}^N$, is a constant. 
\end{definition}
 
Let us  give some examples of algebra wmv. We denote by $AP (\mathbb{R}^{N})$ the space of all Bohr almost periodic functions. The space  $AP (\mathbb{R}^{N})$ is the algebra of functions on $\mathbb{R}^{N}$ that are uniform approximations of finite linear combinations of functions in the set $\{y\mapsto \cos(2\pi k\cdot y), y\mapsto \sin(2\pi k\cdot y),\  k\in\mathbb{R}^{N}\}$.  It is well known that $AP (\mathbb{R}^{N})$ is an ergodic algebra wmv called the almost periodic algebra wmv on $\mathbb{R}^{N}$. We also recall that $\mathcal{C}_{per}(Y)$ the space of continuous $Y=(0,1)^N-$periodic functions on $\mathbb{R}^N$ is an ergodic algebra wmv on $\mathbb{R}^N$. The space $\mathcal{B}_\infty(\mathbb{R}^N)$ of continuous functions on $\mathbb{R}^N$ that converge at infinity is an ergodic algebra wmv. That is the space of all function $u\in\mathcal{B}(\mathbb{R}^N)$ such that $\lim_{|y|\to \infty}u(y)\in\mathbb{R}$. In this case the mean value reduces to  $M(u)=\lim_{|y|\to \infty}u(y)$.

Now, owing to \cite[Theorem~2.2]{NSW}, the following equalities holds: $AP (\mathbb{R}^{N_1}_y)\odot AP (\mathbb{R}^{N_2}_z)=AP (\mathbb{R}^{N_1}_y\times \mathbb{R}^{N_2}_z)$, $\mathcal{C}_{per}(Y)\odot \mathcal{C}_{per}(Z)=\mathcal{C}_{per}(Y\times Z)$ (with $Y=(0,1)^{N_1}$ and $Z=(0,1)^{N_2}$) and $\mathcal{C}_{per}(Y)\odot AP (\mathbb{R}^{N}_z) = \mathcal{C}_{per}(Y; AP (\mathbb{R}^{N}_z))$. Many more examples may be provided, see e.g., \cite[Section\,2.3]{douanlacma} and \cite[Section\,3]{SSW}.

We assume in the sequel that all algebras wmv are ergodic. To the space $B_{A}^{2}(\mathbb{R}^{N})$ we attach the following \textit{corrector} space 
\begin{equation*}
B_{\#A}^{1,2}(\mathbb{R}^{N})=\{u\in W_{loc}^{1,2}(\mathbb{R}^{N}):\nabla
u\in B_{A}^{2}(\mathbb{R}^{N})^{N}\text{ and }M(\nabla u)=0\}\text{.}
\end{equation*}%
Two elements of $B_{\#A}^{1,2}(\mathbb{R}^{N})$ are identify by their
gradients, viz,  $u=v$ in $B_{\#A}^{1,2}(\mathbb{R}^{N})$ if and only if $\nabla (u-v)=0$,
i.e. $\left\Vert \nabla (u-v)\right\Vert _{2}=0$. We may therefore equip $%
B_{\#A}^{1,2}(\mathbb{R}^{N})$ with the gradient norm $\left\Vert
u\right\Vert _{\#,2}=\left\Vert \nabla u\right\Vert _{2}$. This defines a
Banach space \cite[Theorem 3.12]{Casado} containing $B_{A}^{1,2}(\mathbb{R}%
^{N})$ as a subspace.

\subsection{The multiscale convergence} Let $A_y$ (resp. $A_z$) be an ergodic algebra wmv on $\mathbb{R}^N_y$ (resp. $\mathbb{R}^N_z$) for the action $\mathcal{H'}=(H'_\varepsilon)_{\varepsilon>0}$ (resp. $\mathcal{H}''=(H''_\varepsilon)_{\varepsilon>0}$)  of $\mathbb{R}^*_+$ on $\mathbb{R}^N$ given by  $H'_\varepsilon(x)=\frac{x}{\varepsilon}$ (resp. $H''_\varepsilon(x)=\frac{x}{\varepsilon^2}$) and let $A=A_y \odot A_z$  be their product, an algebra wmv on $\mathbb{R}^N\times\mathbb{R}^N$ for the product action $\mathcal{H}=\mathcal{H}'\times \mathcal{H}''$  of $\mathbb{R}^*_+$ on $\mathbb{R}^{2N}=\mathbb{R}^N\times\mathbb{R}^N$ given by $\mathcal{H}=(H_\varepsilon)_{\varepsilon>0}$, with
$$
H_\varepsilon(y,z)=\left( \frac{y}{\varepsilon},\frac{z}{\varepsilon^2} \right) \quad \text{ for } \ y,z\in\mathbb{R}^N \text{ and } \varepsilon>0.
$$

The mean value on $\mathbb{R}^N$ for the actions $\mathcal{H}'$, $\mathcal{H}''$ and $\mathcal{H}$ are respectively denoted by $M_y$, $M_z$ and  $M$. Also, the letter $E$ denote throughout an ordinary sequence of strictly positive real numbers admitting zero as accumulation point. Finally, let $\Omega$ be throughout this section a nonempty open subset of $\mathbb{R}^N$.

\begin{definition}
A sequence $(u_\varepsilon)_{\varepsilon>0} \subset L^2(\Omega)$ is said to weakly multiscale converge in  $L^2(\Omega)$ to some $u_0\in L^2(\Omega; \mathcal{B}^2_A(\mathbb{R}^{2N}))$ if as $E\ni \varepsilon \longrightarrow 0$, we have
\begin{equation*}
\int_{\Omega}u_\varepsilon(x)v(x,\frac{x}{\varepsilon},\frac{x}{\varepsilon^2}) \ dx \longrightarrow \int_{\Omega}M(u(x,-)v(x,-)) \ dx
\end{equation*}
for every $v \in L^{2}(\Omega,A)$, where for a.e. $x\in \Omega$, $v(x,-)(y,z)=v(x,y,z)$ for $(y,z)\in \mathbb{R}^N \times \mathbb{R}^N$. We express this by $u_\varepsilon \xrightarrow{w-ms} u_0$ in $L^2(\Omega)$.
\end{definition}

\begin{definition}
A sequence $(u_\varepsilon)_{\varepsilon>0} \subset L^2(\Omega)$ is said to strongly multiscale converge in  $L^2(\Omega)$ to some $u_0\in L^2(\Omega; \mathcal{B}^2_A(\mathbb{R}^{2N}))$ if as $E\ni \varepsilon \longrightarrow 0$, we have $u_\varepsilon \xrightarrow{w-ms} u_0$ and $\|u_\varepsilon\|_{L^2(\Omega)} \to \|u_0\|_{L^2(\Omega; \mathcal{B}^2_A(\mathbb{R}^{2N}))}$. We express this by $u_\varepsilon \xrightarrow{s-ms} u_0$ in $L^2(\Omega)$.
\end{definition}

Without the following compactness theorems, the multiscale convergence theory would be of no interest. 

\begin{theorem}\label{t2.1}
Let $(u_\varepsilon)_{\varepsilon\in E}$ be a bounded sequence in $L^2(\Omega)$. Then there exist a subsequence $E'$ of E and a function $u\in  L^2(\Omega; \mathcal{B}^2_A(\mathbb{R}^{2N}))$ such that $u_\varepsilon \xrightarrow{w-ms} u_0$ in $L^2(\Omega)$ as $E'\ni\varepsilon\to 0$.
\end{theorem}
\begin{theorem}\label{t2.2}
Let $(u_\varepsilon)_{\varepsilon\in E}$ be a bounded sequence in $H^1(\Omega)$. Then there exist a subsequence $E'$ of $E$ and a triple \textbf{u}= $(u_0,u_1,u_2)\in H^1(\Omega) \times L^2(\Omega; B^{1,2}_{\# A_y}(\mathbb{R}^{N})) \times L^2(\Omega; \mathcal{B}^2_{A_y}(\mathbb{R}_y^{N}; B^{1,2}_{\# A_z}(\mathbb{R}^{N})))$  such that, as $E'\ni \varepsilon \longrightarrow 0$,
\begin{equation}\label{eq2.6}
 u_{\varepsilon} \longrightarrow u_0 \text{\ \  } in \text{\  \ } H^1(\Omega)-weak
\end{equation}
and
\begin{equation}\label{eq2.7}
 \frac{\partial u_\varepsilon}{\partial x_j} \xrightarrow{w-ms} \frac{\partial u_0}{\partial x_j}+\frac{\partial u_1}{\partial y_j}+\frac{\partial u_2}{\partial z_j} \text{\ \ }in \text{\ \ } L^2(\Omega)\ \  \;\; (1 \leq j \leq N).
\end{equation}
\end{theorem}

\section{Estimates and passage to the limit}
The homogenization procedure starts with a boundedness result for the sequence $(u_\varepsilon)_{\varepsilon> 0}$. The  passage to the limit requires a structural hypothesis. In addition to (\ref{eq001}) and (\ref{eq11.2})-(\ref{eq11.4}), we assume that for almost all $x\in \Omega$
\begin{equation}\label{eq33.1}
A(x)\in \left( B^2_{A_y\odot A_z} \right)^{N\times N},
\end{equation}
where $A_y$ and $A_z$ are ergodic algebras wmv on $\mathbb{R}^N_y$ and $\mathbb{R}^N_z$, respectively, and $A_y\odot A_z$ is the product algebra wmv of $A_y$ and $A_z$.
\begin{proposition}
The sequence $(u_\varepsilon)_{\varepsilon> 0}$ is bounded in $H^1_0(\Omega)$.
\end{proposition}  
\begin{proof}  
Let $\psi^+_\varepsilon=max(\psi_\varepsilon;0)=\frac{1}{2}(\psi_\varepsilon+|\psi_\varepsilon|)$, then $\psi^+_\varepsilon\in K_\varepsilon$. Moreover, since $\psi_\varepsilon$ converges strongly in $H^1(\Omega)$, there exists a constant $C>0$ such that $||\psi^+_\varepsilon||_{H^1_0(\Omega)} \leq C$. Taking $v=\psi^+_\varepsilon$ in (\ref{eq1.1}), we get
\begin{equation*}
   \displaystyle{\int_{\Omega}A(x, \frac{x}{\varepsilon},\frac{x}{\varepsilon^2})\nabla u_\varepsilon(x)\cdot\nabla (\psi^+_\varepsilon(x)- u_\varepsilon(x))dx} \geq \displaystyle{\int_{\Omega} f(x)(\psi^+_\varepsilon(x)- u_\varepsilon(x))dx},
\end{equation*}               
from which we deduce
\begin{equation*}
   \displaystyle{\int_{\Omega}A(x, \frac{x}{\varepsilon},\frac{x}{\varepsilon^2})\nabla u_\varepsilon\cdot\nabla u_\varepsilon dx} \leq \displaystyle{ \int_{\Omega}A(x, \frac{x}{\varepsilon},\frac{x}{\varepsilon^2})\nabla u_\varepsilon\cdot\nabla \psi^+_\varepsilon dx+\int_{\Omega} f(u_\varepsilon-\psi^+_\varepsilon)dx}. 
\end{equation*}          
Thus
 \begin{eqnarray*}
    \displaystyle{\alpha||\nabla u_\varepsilon||_{L^2(\Omega)^N}^2} &\leq& \displaystyle{ \int_{\Omega}A(x, \frac{x}{\varepsilon},\frac{x}{\varepsilon^2})\nabla u_\varepsilon\cdot\nabla \psi^+_\varepsilon dx+\int_{\Omega} f(u_\varepsilon-\psi^+_\varepsilon)dx} \\
    &\leq& ||f||_{L^2(\Omega)}\left(||\psi^+_\varepsilon||_{L^2(\Omega)}+||u_\varepsilon||_{L^2(\Omega)}\right)+C||\nabla u_\varepsilon||_{L^2(\Omega)^N}||\nabla \psi^+_\varepsilon||_{L^2(\Omega)^N}\\
    &\leq& ||f||_{L^2(\Omega)}\left(C_1+||u_\varepsilon||_{L^2(\Omega)}\right)+C||\nabla u_\varepsilon||_{L^2(\Omega)^N}\\
    &\leq& ||f||_{L^2(\Omega)}\left(C_1+C_2||\nabla u_\varepsilon||_{L^2(\Omega)}\right)+C||\nabla u_\varepsilon||_{L^2(\Omega)^N} \ \ \ \\
    &\leq& C_1||f||_{L^2(\Omega)}+C_2||f||_{L^2(\Omega)}||\nabla u_\varepsilon||_{L^2(\Omega)}+C||\nabla u_\varepsilon||_{L^2(\Omega)^N} \\
    &\leq& C_1||f||_{L^2(\Omega)}+C_2\left[\frac{\delta^2}{2}||f||_{L^2(\Omega)}^2+\frac{\delta^{-2}}{2}||\nabla u_\varepsilon||_{L^2(\Omega)^N}^2\right]\\
    &&+C\left[\frac{\theta^{2}}{2}+\frac{\theta^{-2}}{2}||\nabla u_\varepsilon||_{L^2(\Omega)^N}^2 \right] \ \ \ \ \ (\mbox{Young's inequality with} \ (\delta, \theta)\in(\mathbb{R}^*_+)^2)\\
        &\leq& C_2\frac{\delta^{-2}}{2}||\nabla u_\varepsilon||_{L^2(\Omega)^N}^2+C \frac{\theta^{-2}}{2}||\nabla u_\varepsilon||_{L^2(\Omega)^N}^2 +C_2\frac{\delta^2}{2}||f||_{L^2(\Omega)}^2\\
        &&+ C_1 ||f||_{L^2(\Omega)}+ C \frac{\theta^{2}}{2}.
 \end{eqnarray*}  
Taking $\delta=\sqrt{\frac{2C_2}{\alpha}}$  and $\theta=\sqrt{\frac{2C}{\alpha}}$, we obtain
 \begin{equation*}
 \displaystyle{\alpha||\nabla u_\varepsilon||_{L^2(\Omega)^N}^2} \leq \frac{\alpha}{4}||\nabla u_\varepsilon||_{L^2(\Omega)^N}^2+\frac{\alpha}{4}||\nabla u_\varepsilon||_{L^2(\Omega)^N}^2+ K, 
 \end{equation*}
where   $K=C_1||f||_{L^2(\Omega)}+\frac{C_2^2}{\alpha}||f||_{L^2(\Omega)}^2+\frac{C^2}{\alpha}$.  
Hence, we are led to
 $$||\nabla u_\varepsilon||_{L^2(\Omega)^N} \leq \sqrt{\frac{2K}{\alpha}},$$
 and the proof is completed.
\end{proof}
     
The dependence on $\varepsilon$ of the closed convex set $K_\varepsilon$ appearing in the problem (\ref{eq1.1}) is a  problem with respect to the limit passage. The problem (\ref{eq1.1}) need to be reformulated. Let us introduce the following space $$K=\{v\in H^1(\Omega); v \geq 0  \ \ a.e \ \mbox{in}\ \Omega \}.$$
It is straightforward that $u_\varepsilon \in K_\varepsilon$ is a solution to (\ref{eq1.1}) if and only if $\hat{u}_\varepsilon =u_\varepsilon - \psi_\varepsilon \in K$ solves the  following problem:
\begin{eqnarray} \label{eqt2}
            \left \{ \begin{array}{lll}
            \mbox{Find} \ \hat{u}_\varepsilon \in \ K\ \ \mbox{such that} \\ 
          \displaystyle{\int_{\Omega}A^\varepsilon(x)\nabla (\hat{u}_\varepsilon(x)+\psi_\varepsilon(x))\cdot\nabla (\hat{v}_\varepsilon(x)-\hat{ u}_\varepsilon(x))dx} &\geq& \displaystyle{\int_{\Omega} f(x)(\hat{v}_\varepsilon(x)- \hat{u}_\varepsilon(x))dx}\\
          \mbox{for all} \ \hat{v}_\varepsilon \in K.
              \end{array}\right.
            \end{eqnarray} 
We need the following spaces in the sequel: 
$$
\mathbb{V} := H^1(\Omega) \times L^2(\Omega; B^{1,2}_{\# A_y}(\mathbb{R}^{N})) \times L^2(\Omega; \mathcal{B}^p_{A_y}(\mathbb{R}_y^{N}; B^{1,2}_{\# A_z}(\mathbb{R}^{N}))),
$$
$$
 K_2=\{(v_0,v_1,v_2) \in \mathbb{V}; \ v_0 \geq 0  \ \  a.e \ \ \mbox{in} \ \ \Omega \}
$$
 and 
$$
  K_3=\{(v_0,v_1,v_2) \in \mathbb{V};v_0 \geq \psi_0  \ \  a.e \ \ \mbox{in} \ \ \Omega \}.
  $$

\begin{theorem}
There exist a triple $(u_0,u_1,u_2)\in\mathbb{V}$ 
and a subsequence $E'$ of $E$ such that as $E' \ni \varepsilon \longrightarrow 0$,\begin{eqnarray}
u_\varepsilon  & {\longrightarrow} &  u_0 \qquad \qquad  \qquad \qquad \qquad  \text{in} \ \ H^1(\Omega)\mbox{-weak},\label{eqt4}\\
\label{eqt5}
 \nabla u_\varepsilon & \xrightarrow{w-ms} &  \nabla u_0 + \nabla_y u_1 + \nabla_z u_2 \ \ \ \ \mbox{in} \ \ L^2(\Omega)^N. 
 \end{eqnarray}
  Moreover, the triple $(u_0,u_1,u_2)\in\mathbb{V}$ is the unique solution to the variational inequality
  
 \begin{eqnarray} \label{eqt12}
                        \left \{ \begin{array}{lll}
                        \mbox{Find} \ (u_0,u_1,u_2) \in \ K_3 \ \ \mbox{such that} \\ 
                      \displaystyle{\int_{\Omega}\!} M[A(\nabla u_0 +\nabla_yu_1+\!\nabla_zu_2 )\cdot(\nabla(v_0-u_0) \!+\!\nabla_y(v_1-u_1)\!+\!\nabla_z(v_2-u_2))]dx \\ \ \ \ \ \ \ \ \ \ \ \ \ \ \ \ \ \ \  \geq \displaystyle{\int_{\Omega} f(v_0-u_0) dx } \\
                      \mbox{for all} \ (v_0,v_1,v_2) \in K_3.
                          \end{array}\right.
                        \end{eqnarray}      
\end{theorem}
\begin{proof}
The sequence $(\hat{u}_\varepsilon)_{\varepsilon\in E}$ is bounded in $H^1(\Omega)$ because  $(u_\varepsilon)_{\varepsilon\in E}$ is bounded in $H^1(\Omega)$ and  $(\psi_\varepsilon)_{\varepsilon\in E}$ converges strongly in $H^1(\Omega)$. Therefore, according to Theorem~\ref{t2.2}, there exist a triple $ (\hat{u}_0,\hat{u}_1,\hat{u}_2) \in \mathbb{V}$ and a subsequence $E'$ of $E$ such that
\begin{eqnarray*}
\hat{u}_\varepsilon  & {\longrightarrow} & \hat{u}_0 \quad\  \qquad  \qquad \qquad \qquad  \text{in} \ \ H^1(\Omega)\mbox{-weak},\\
 \nabla \hat{u}_\varepsilon & \xrightarrow{w-ms} & \nabla \hat{u}_0 + \nabla_y\hat{u}_1+\nabla_z\hat{u}_2 \ \ \ \ \mbox{in} \ \ L^2(\Omega)^N. 
 \end{eqnarray*}
Indeed, $\hat{u}_0\geq 0$ almost everywhere in $\Omega$ so that $(\hat{u}_0,\hat{u}_1,\hat{u}_2)\in K_2$. Now let $\hat{v}_0 \in \mathcal{C}^\infty_0(\Omega), \ \hat{v}_1 \in \mathcal{C}^\infty_0(\Omega)\otimes A^\infty_y \mbox{and} \ \hat{v}_2 \in \mathcal{C}^\infty_0(\Omega)\otimes A^\infty,$ with $\hat{v}_0 \geq 0$, and define
\begin{equation*}
 \hat{v}_\varepsilon(x)=\hat{v}_0(x)+\varepsilon\hat{v}_1(x,\frac{x}{\varepsilon})+\varepsilon^2\hat{v}_2(x,\frac{x}{\varepsilon},\frac{x}{\varepsilon^2}), \ \ \ \ \ \ (x\in \Omega, \varepsilon >0).
 \end{equation*}
We assume that $\hat{v}_\varepsilon(x) \geq 0$ for all $x\in \Omega$, if not we may consider $\varepsilon$ small enough to have it. We have 
\begin{eqnarray*}
 \nabla\hat{v}_\varepsilon(x)&=&\nabla \hat{v}_0(x)+\varepsilon\nabla_x\hat{v}_1(x,\frac{x}{\varepsilon})+ \nabla_y\hat{v}_1(x,\frac{x}{\varepsilon})+\varepsilon^2\nabla_x\hat{v}_2(x,\frac{x}{\varepsilon},\frac{x}{\varepsilon^2})\\
 &&+\varepsilon\nabla_y\hat{v}_2(x,\frac{x}{\varepsilon},\frac{x}{\varepsilon^2})+\nabla_z\hat{v}_2(x,\frac{x}{\varepsilon},\frac{x}{\varepsilon^2})
 \end{eqnarray*}  
 and we recall that as $\varepsilon\to 0$ the following convergence takes place 
 \begin{equation*}
\nabla \hat{v}_\varepsilon  \xrightarrow{s-ms} \nabla \hat{v}_0+\nabla_y\hat{v}_1+\nabla_z\hat{v}_2 \ \ \ \ \ \ \ \mbox{in} \ \ L^2(\Omega)^N. 
 \end{equation*}
Passing to the limit as $E'\ni\varepsilon \longrightarrow 0$ in problem (\ref{eqt2}) yields \begin{eqnarray*} \label{eqt10}
            \left \{ \begin{array}{lll}
           (\hat{u}_0,\hat{u}_1,\hat{u}_2) \in \ K_2 \ \ \mbox{:} \\ 
          \displaystyle{\int_{\Omega}}M[A(\nabla \hat{u}_0 +\nabla_y\hat{u}_1+\nabla_z\hat{u}_2+ \nabla \psi_0)\cdot(\nabla(\hat{v}_0-\hat{u}_0)   +\nabla_y(\hat{v}_1-\hat{u}_1)+\nabla_z(\hat{v}_2-\hat{u}_2))]\ dx \\ \ \ \ \ \ \ \ \ \ \ \ \ \ \ \ \ \ \ \ \ \ \  \geq \displaystyle{\int_{\Omega} f(\hat{v}_0-\hat{u}_0) dx } \\
       \mbox{for all} \ (\hat{v}_0,\hat{v}_1,\hat{v}_2) \in K_2.
              \end{array}\right.
            \end{eqnarray*}  
But, the change of variables $u_0=\hat{u}_0 + \psi_0$, $u_1 = \hat{u}_1$ and $u_2=\hat{u}_2$ show that the above problem appears to be equivalent to the problem (\ref{eqt12}).
Hence, the triple $(u_0,u_1,u_2)\in\mathbb{V}$ is a solution to the variational inequality (\ref{eqt12}). By means of the Stampacchia's lemma it is actually its unique solution.
\end{proof}
\section{Main result: Macroscopic problem}
In order to derive the macroscopic problem, we need to formulate the microscopic and mesoscopic ones.  

\subsection{Microscopic problem}
Taking  $v_0=u_0$ and $v_1=u_1$in (\ref{eqt12}), we get 
\begin{eqnarray*}\label{eqt14}
\left \{ \begin{array}{lll}
 \displaystyle {\int_{\Omega}M\left[ A(x,y,z)(\nabla u_0 +\nabla_yu_1+\nabla_zu_2 )\cdot\nabla_zv_2\right]\ dx \geq  0}\\
  \mbox{for all } \ \  v_2 \in L^2(\Omega; \mathcal{B}^2_{A_y}(\mathbb{R}_y^{N}; B^{1,2}_{\# A_z}(\mathbb{R}^{N}))).
                \end{array}\right.
\end{eqnarray*}
Since $ L^2(\Omega; \mathcal{B}^2_{A_y}(\mathbb{R}_y^{N}; B^{1,2}_{\# A_z}(\mathbb{R}^{N})))$ is a vector space, we obtain the following variational equation 
\begin{eqnarray}\label{eqt14}
\left \{ \begin{array}{lll}
 \displaystyle {\int_{\Omega}M\left[ A(x,y,z)(\nabla u_0 +\nabla_yu_1+\nabla_zu_2 )\cdot\nabla_zv_2\right]\ dx =  0}\\
  \mbox{for all } \ \  v_2 \in L^2(\Omega; \mathcal{B}^2_{A_y}(\mathbb{R}_y^{N}; B^{1,2}_{\# A_z}(\mathbb{R}^{N}))).
   \end{array}\right.
\end{eqnarray}
Now we take $v_2= \varphi \otimes
\omega$ with $\varphi \in \mathcal{C}^\infty_0(\Omega) \otimes A^\infty_y$ and $\omega \in A^\infty_z$, and realize that for almost all $(x,y)\in\Omega\times \mathbb{R}^N$, the function $u_2(x,y)\in B^{1,2}_{\# A_z}(\mathbb{R}^N)$ solves the following problem
\begin{eqnarray}\label{eqt16}
               \left \{ \begin{array}{lll}
                M_z\left[ A(x,y,z)\nabla_zu_2 \cdot\nabla_z\omega\right] = -M_z\left[ A(x,y,z)(\nabla u_0 +\nabla_yu_1 )\cdot\nabla_z\omega\right]\\
                 \mbox{for all } \ \  \omega \in A^\infty_z.
                               \end{array}\right.
\end{eqnarray}
This being so, for almost all $(x,y)\in\Omega\times \mathbb{R}^N$ we introduce the following microscopic problem:
\begin{eqnarray} \label{eqt17}
            \left \{ \begin{array}{lll}
           \chi^j \equiv \chi^j(x,y)  \in  B^{1,2}_{\# A_z}(\mathbb{R}^{N}) \ \ \mbox{such that} \\ 
          M_z\left[ A(x,y,z)\nabla_z\chi^j \cdot\nabla_z\omega\right] = -M_z\left[\sum_{k=1}^{N} a_{jk}\frac{\partial \omega}{\partial z_k}\right] \\
          \mbox{for all} \ \omega \in B^{1,2}_{\# A_z}(\mathbb{R}^{N})
              \end{array}\right.
            \end{eqnarray}  
which possesses a unique solution. Setting $\chi =(\chi_j)_{1 \leq j \leq N}$, it is easy to check that the function $(x,y,z)\mapsto \chi(z)(\nabla u_0(x)+\nabla_y u_1(x,y))$ is a solution to (\ref{eqt16}).  Therefore by uniqueness of the solution to (\ref{eqt16}), it holds allmost everywhere in $\Omega \times \mathbb{R}^{N}\times\mathbb{R}^N$ that $u_2 =\chi\cdot(\nabla u_0+\nabla_yu_1)$. Letting $\displaystyle{(\nabla_z\chi)_{ij}=\frac{\partial\chi^j}{\partial z_i} \ \ (1 \leq i,j\leq N)}$, we may then write
   \begin{equation}\label{eqt00}
     \nabla_zu_2(x,y,z)=\left[\nabla_z\chi(z)\right]\left(\nabla u_0(x)+\nabla_yu_1(x,y)\right)         
     \end{equation}
for almost every $(x,y,z) \in \Omega\times \mathbb{R}^{N}\times\mathbb{R}^N$.
We are now in a position to upscale to the mesoscopic scale.     
\subsection{Mesoscopic problem}
Taking $v_0=u_0$ and $v_2=u_2$ in (\ref{eqt12}), we are led to 
\begin{eqnarray}\label{eqt19}
\left \{ \begin{array}{lll}
 \displaystyle {\int_{\Omega}M\left[ A(x,y,z)(\nabla u_0 +\nabla_yu_1+\nabla_zu_2 )\cdot\nabla_zv_1\right]\ dx = 0}\\
  \mbox{for all } \ \  v_1 \in L^2(\Omega;  B^{1,2}_{\# A_y}(\mathbb{R}^{N})).
                \end{array}\right.
\end{eqnarray}
Choosing $v_1= \varphi \otimes
 \omega$ with $\varphi \in \mathcal{C}^\infty_0(\Omega)$ and $\omega \in A^\infty_y$ in (\ref{eqt19}), it appears that for almost all $x\in \Omega$, $u_1(x)$ is a solution to the following problem 
 \begin{eqnarray}\label{eqt20}
 \left \{ \begin{array}{lll}
 u_1(x) \in B^{1,2}_{\# A_y}(\mathbb{R}^{N}),\\
  \displaystyle {M[ A(x,y,z)(I+\nabla_z \chi)(\nabla u_0 +\nabla_yu_1)\cdot\nabla_y\omega] = 0}\\
   \mbox{for all } \ \  \omega \in A^\infty_y,
                 \end{array}\right.
              \end{eqnarray}
         which also writes 
          \begin{eqnarray}\label{eqt21}
               \left \{ \begin{array}{lll}
                u_1(x) \in B^{1,2}_{\# A_y}(\mathbb{R}^{N}),\\
                M_y[ \tilde{A}(x,y)\nabla_yu_1 \cdot\nabla_y\omega] = -M_y[ \tilde{A}(x,y)\nabla u_0\cdot\nabla_y\omega]\\
                 \mbox{for all } \ \  \omega \in A^\infty_y.
                               \end{array}\right.
                            \end{eqnarray}
Where $\tilde{A}(x,y)=M_z[A(x,y,z)(I+\nabla_z \chi)]$ is the well-known symmetric positive-definite  averaged matrix. For almost all $x\in \mathbb{R}^N$, we introduce the  mesoscopic problem:
\begin{eqnarray} \label{eqt22}
            \left \{ \begin{array}{lll}
            \mbox{Find} \ \theta^j  \in  B^{1,2}_{\# A_y}(\mathbb{R}^{N}) \ \ \mbox{such that} \\ 
          M_y[ \tilde{A}\nabla_y\theta^j \cdot\nabla_y\omega] = -M_y[\sum_{k=1}^{N} \tilde{a}_{kj}\frac{\partial \omega}{\partial y_k}], \\
          \mbox{for all} \ \omega \in B^{1,2}_{\# A_y}(\mathbb{R}^{N}).
              \end{array}\right.
            \end{eqnarray}  
and recall that it possesses a unique solution. It is not difficult to check that the function $(x,y)\mapsto \theta (y) \nabla u_0(x)$ is also a solution to (\ref{eqt21}), so that almost everywhere in $\Omega \times \mathbb{R}^N$, it holds by uniqueness of the solution to (\ref{eqt21}) that $ u_1(x,y)=\theta(y)\nabla u_0(x)$. With the matrix notation, $   \displaystyle{(\nabla_y\theta)_{ij}=\frac{\partial\theta^j}{\partial y_i}, \ \ (1 \leq i,j\leq N)}$, it follows that
\begin{equation}\label{eqt23}
    \displaystyle{ \nabla_yu_1(x,y)=\nabla_y\theta(y)\nabla u_0(x)} \mbox{\ \ \ \ \  a.e  in \ }   (x,y)\in \Omega \times \mathbb{R}^N       
     \end{equation}
\subsection{Homogenization result: Macroscopic problem}
Let $K_0=\{ H^1_0(\Omega), v_0\geq \psi_0  \ a.e \ \mbox{in} \ \Omega\}$. Choosing $v_2=u_2$ and $v_1=u_1$ in (\ref{eqt12}) leads to 
\begin{eqnarray}\label{eqt26}
  \left \{ \begin{array}{lll}   
\displaystyle{  \int_{\Omega}M[ A(\nabla u_0 +\nabla_yu_1+\nabla_zu_2)\cdot\nabla (v_0-u_0)]\ dx  \geq \int_{\Omega}f(v_0-u_0) \ dx},\\
\mbox{for all}\ v_0\in K_0,
    \end{array}\right. 
  \end{eqnarray}
which rewrites using (\ref{eqt00}) and (\ref{eqt23}) as
 \begin{eqnarray}\label{eqt27}
   \left \{ \begin{array}{lll}   
 \displaystyle{  \int_{\Omega}M[ A(I+ \nabla_z \chi)(I+\nabla_y \theta)\nabla u_0\cdot\nabla (v_0-u_0)]\ dx  \geq \int_{\Omega}f(v_0-u_0) \ dx},\\
 \mbox{for all}\ v_0\in K_0.
     \end{array}\right. 
   \end{eqnarray} 
It is classical that the matrix $ M[A(I+\nabla_z \chi)(I+ \nabla_y \theta)]$ is symmetric and positive-definite\,\cite{BLP}. Thus, the problem (\ref{eqt27}) admits a unique solution so that the whole fundamental sequence $(u_\varepsilon)_{\varepsilon\in E}$ weakly converges to $u_0$ in $H^1_0(\Omega)$. The arbitrariness of the fundamental sequence $E$ implies the convergence of the generalized sequence $(u_\varepsilon)_{\varepsilon>0}$. This end the proof of our main result.

\begin{theorem}
For any $\varepsilon >0$, let $u_\varepsilon$ be the unique solution to the variational inequality (\ref{eq1.1}) with the hypotheses (\ref{eq001}), (\ref{eq11.2})-(\ref{eq11.4}) and (\ref{eq33.1}). Then as $0< \varepsilon \longrightarrow 0$, we have $u_\varepsilon \longrightarrow u_0$ in $H^1_0(\Omega)$-weak, and strongly in $L^2(\Omega)$, where $u_0$ is the unique solution to the following homogenized variational inequality
\begin{eqnarray}
     \left \{ \begin{array}{lll} 
     u_0\in K_0,\\ 
  \displaystyle{   \int_{\Omega}A^{*}(x)\nabla u_0\cdot\nabla (v_0-u_0)\ dx  \geq \int_{\Omega}f(v_0-u_0) \ dx },\\
  \mbox{for all}\ v_0\in K_0,
   \end{array}\right.
   \end{eqnarray} 
where, for almost every $x\in \Omega$, the symmetric positive-definite homogenized matrix is given by
$$
A^{*}(x)=M[A(I+\nabla_z \chi)(I+ \nabla_y \theta)],
$$
the functions $\chi$ and $\theta$ being the solutions to the microscopic problem (\ref{eqt17}) and the mesoscopic problem (\ref{eqt22}), respectively.
\end{theorem}

\end{document}